\documentclass[aps,prl,reprint,superscriptaddress,footinbib]{revtex4-2}

\usepackage{tensor}
\usepackage{amsmath}
\usepackage{amsthm}
\usepackage{amssymb}
\usepackage{graphicx}
\usepackage{comment}
\usepackage{amsfonts}
\usepackage{bm}
\usepackage{braket}
\usepackage{xcolor}

\usepackage{thmtools}
\usepackage{tikz}
\usetikzlibrary{arrows.meta,positioning,calc,fit}

\definecolor{steelblue}{HTML}{5083C5}
\definecolor{medgreen}{HTML}{397A3E}

\usepackage[colorlinks=true,linkcolor=medgreen,urlcolor=teal,citecolor=steelblue,hypertexnames=false]{hyperref}

\usepackage{enumerate}
\usepackage{mathtools}
\usepackage{bbm}

\usepackage{mathrsfs}
\usepackage{dsfont}

\newcommand{\ii}{\mathrm{i}}

\newtheorem{theorem}{Theorem}

\newtheorem*{lemma*}{Lemma}

\DeclareMathOperator{\Tr}{Tr}
\DeclareMathOperator{\diag}{diag}

\newcommand{\SLD}{\mathrm{SLD}}
\newcommand{\Holevo}{\mathrm{H}}
\newcommand{\HS}{\mathrm{HS}}

\begin{document}

\title{Exact Characterization of the Holevo Bound by a Quantum Fisher Information Family}

\author{Koji Yamaguchi}
\email{koji.yamaguchi@uwaterloo.ca}
\affiliation{Department of Physics, University of Waterloo, Waterloo, ON N2L 3G1, Canada}
\affiliation{Perimeter Institute for Theoretical Physics, Waterloo, Ontario N2L 2Y5, Canada}

\author{Hiroyasu Tajima}
\email{hiroyasu.tajima@inf.kyushu-u.ac.jp}
\affiliation{Department of Informatics, Faculty of Information Science and Electrical Engineering,
Kyushu University, 744 Motooka, Nishi-ku, Fukuoka, 819-0395, Japan}
\affiliation{JST, FOREST, 4-1-8 Honcho, Kawaguchi, Saitama, 332-0012, Japan}

\begin{abstract}
The quantum Cram\'er--Rao bound constrains the precision of parameter estimation through the symmetric logarithmic derivative quantum Fisher information (SLD QFI). It is asymptotically achievable for regular single-parameter estimation models, but generally not in the multiparameter setting, where optimal measurements for different parameters may be incompatible. For multiparameter estimation, allowing collective measurements, the corresponding asymptotically achievable precision limit is the Holevo bound, whose standard formulation is an optimization over auxiliary Hermitian operators. We bridge this apparent difference in formulation by giving an exact characterization of the Holevo bound in terms of a family of QFIs interpolating between the SLD and right logarithmic derivative (RLD) QFIs. 
Specifically, for any locally identifiable finite-dimensional estimation problem, we prove that the operator-feasible region in the standard variational definition of the Holevo bound coincides with the intersection of the regions defined by the corresponding Cram\'er--Rao-type matrix constraints for the entire QFI family. This characterization also applies to rank-deficient states and does not require any prior choice of weight matrix. Consequently, the conventional weight-dependent Holevo bound is recovered by minimizing the corresponding weighted cost over the resulting common feasible region. 
\end{abstract}

\maketitle

Quantum metrology~\cite{giovannetti_QuantumEnhancedMeasurementsBeatingStandardQuantum_2004,giovannettiQuantumMetrology2006,pezze_QuantumMetrologyNonclassicalstatesatomic_2018} studies the estimation of physical parameters encoded in quantum systems and the role of quantum effects in improving measurement precision. Its applications span a broad range of precision-measurement settings, including atomic clocks~\cite{huelga_ImprovementFrequencyStandardsQuantumEntanglement_1997,ludlow_OpticalAtomicClocks_2015}, magnetometers~\cite{budker_OpticalMagnetometry_2007}, gravitational-wave detectors~\cite{caves_QuantummechanicalNoiseInterferometer_1981,tse_QuantumEnhancedAdvancedLIGODetectorsEra_2019_at}, and quantum imaging~\cite{tsang_QuantumTheorySuperresolutionTwoIncoherent_2016,albarelli_PerspectiveMultiparameterQuantummetrologytheoretical_2020}. A central aim of quantum metrology is to identify the properties of quantum states that determine the ultimate precision attainable in parameter estimation.

For single-parameter estimation, the symmetric logarithmic derivative quantum Fisher information (SLD QFI) provides such a characterization. The quantum Cram\'er--Rao bound (QCRB)~\cite{helstrom_MinimumMeansquaredErrorestimatesquantum_1967,helstromQuantumDetectionEstimation1969,braunstein_StatisticalDistanceGeometryquantumstates_1994} states that the mean-squared error (MSE) of any locally unbiased estimator is bounded from below by the inverse of the SLD QFI. Moreover, in the single-parameter setting, the bound is asymptotically attainable in the many-copy limit under suitable regularity conditions~\cite{braunstein_StatisticalDistanceGeometryquantumstates_1994,barndorff-nielsenFisherInformationQuantum2000}. The SLD QFI thus serves as a fundamental figure of merit for single-parameter quantum estimation.

The situation changes fundamentally in multiparameter estimation. While the SLD QFI matrix provides a direct generalization of the QCRB, this bound is generally not attainable, reflecting the possible incompatibility of measurements associated with different parameters~\cite{yuenMultipleparameterQuantumEstimation1973,ragyCompatibilityMultiparameterQuantum2016,albarelli_PerspectiveMultiparameterQuantummetrologytheoretical_2020}. Instead, the Holevo bound~\cite{holevo_ProbabilisticStatisticalAspectsquantumtheory_1982,nagaoka_newapproachtoCR_1989} (also referred to as the Holevo Cram\'er--Rao bound or the Holevo--Nagaoka bound) provides a lower bound on the weighted MSE and is asymptotically attainable in the many-copy limit when collective measurements are allowed~\cite{hayashi_AsymptoticPerformanceOptimalstateestimation_2008,yamagataQuantumLocalAsymptotic2013,yangAttainingUltimatePrecision2019}. Thus, the Holevo bound serves in multiparameter estimation as the counterpart of the QCRB in single-parameter estimation.

Despite their analogous roles in setting precision limits, the QCRB and the Holevo bound appear to arise from rather different mathematical structures. The QCRB is governed by the SLD QFI, which has a geometric interpretation in terms of the local distinguishability of quantum states. In contrast, the Holevo bound is formulated as a variational optimization over Hermitian operators subject to local unbiasedness constraints, which does not by itself provide an immediate geometric characterization. Over the decades since the pioneering works in the 1970s and 1980s, a variety of approaches have been developed to gain a deeper understanding of the Holevo bound. 
One line of research has reformulated the Holevo bound within the framework of convex optimization, using semidefinite programming~\cite{albarelli_EvaluatingHolevoCramerRaoBoundMultiparameter_2019} and more general conic programming~\cite{hayashi_TightCramerRaoTypeboundsmultiparameter_2023,hayashi_FindingOptimalProbestatemultiparameter_2024}, thereby clarifying its optimization structure and computational tractability.
Another line of research has developed a convex-geometric characterization of the Holevo bound in terms of the admissible region of covariance matrices~\cite{Gill_Conciliation_2008,gillAsymptoticQuantumStatistical2013}.
A complementary line of research aims to understand the Holevo bound through quantum Fisher geometry. 
For certain classes of estimation problems, the Holevo bound coincides with either the SLD or right-logarithmic-derivative (RLD) Cram\'er--Rao bound~\cite{holevo_ProbabilisticStatisticalAspectsquantumtheory_1982,suzuki_InformationGeometricalCharacterizationQuantumStatistical_2019,albarelli_PerspectiveMultiparameterQuantummetrologytheoretical_2020}. Building on this connection, a one-parameter family of logarithmic derivatives interpolating between the SLD and RLD, termed the $\beta$-logarithmic derivative ($\beta$-LD), was introduced in Ref.~\cite{yamagataMaximumLogarithmicDerivative2021}. Each $\beta$-LD induces a corresponding QFI matrix and a Cram\'er--Rao-type bound. Maximizing these bounds over $\beta$ for a given weight matrix defines the maximum logarithmic derivative bound, which coincides with the Holevo bound for a class of two-parameter settings~\cite{yamagataMaximumLogarithmicDerivative2021}. 
Using an analytic expression for the Holevo bound in two-parameter qubit models~\cite{suzuki_ExplicitFormulaHolevoboundtwoparameter_2016}, the admissible region of covariance matrices was identified with the joint constraints imposed by the entire family of $\beta$-LD QFI, providing a weight-independent characterization of the Holevo bound for the two-parameter qubit models~\cite{niu_HolevoBoundIndependentweightmatrices_2024}, where possible generalizations beyond this setting were also conjectured. Despite these advances, a general information-geometric characterization of the Holevo bound, analogous to the role played by the SLD QFI for the QCRB in single-parameter estimation, has yet to be established.

In this Letter, we prove that the operator-feasible region in the standard variational definition of the Holevo bound coincides exactly with the intersection of the regions defined by the matrix inequalities associated with the entire family of $\beta$-LD QFIs. This gives an equivalent information-geometric formulation of the Holevo optimization entirely in terms of the QFI matrices. Consequently, the Holevo bound is recovered by imposing the full family of Cram\'er--Rao-type matrix constraints simultaneously and minimizing the weighted cost over their common feasible region.

Our proof first recasts the operator-feasibility condition in the standard definition of the Holevo bound into the real-coefficient form underlying the semidefinite-programming formulation~\cite{albarelli_EvaluatingHolevoCramerRaoBoundMultiparameter_2019}. We then apply the matrix feasibility theorem in Ref.~\cite{yamaguchi_QuantifyingSymmetryBreakingMetricAdjusted_2026} to eliminate the auxiliary variables and identify the resulting inequalities with the joint matrix constraints generated by the entire family of $\beta$-LD QFIs. This establishes the equality between the operator-feasible region and the QFI-defined region determined by the interpolating family. Consequently, the conventional weight-dependent Holevo bound is recovered as the infimum of the weighted cost over their common feasible region. Together with the characterization of quantum Gaussian shift-model convertibility in Ref.~\cite{yamaguchi_QuantifyingSymmetryBreakingMetricAdjusted_2026}, our result reveals a common information-geometric structure underlying asymptotic multiparameter quantum estimation and Gaussian shift-model convertibility: both are characterized by the same family of QFIs interpolating between the SLD and RLD.

\textit{\textbf{Quantum Cram\'er--Rao bound.}}--- Consider a differentiable $m$-parameter quantum statistical model $\{\rho_\theta\}_{\theta\in\Theta}$ on a finite-dimensional Hilbert space $\mathcal{H}$, where $\Theta\subset\mathbb{R}^m$ is open. Fix $\theta_0\in\Theta$, and write $\rho\coloneqq\rho_{\theta_0}$ and $\partial_i\rho\coloneqq\partial_{\theta_i}\rho_\theta|_{\theta=\theta_0}$. A positive operator-valued measure (POVM) $\{\Pi_\omega\}$ induces the probability distribution $p_\theta(\omega)\coloneqq\Tr(\rho_\theta\Pi_\omega)$. For an estimator $\hat\theta(\omega)$, its mean squared error (MSE) matrix at $\theta_0$ is
\begin{align}
    (\Sigma_\rho)_{ij}\coloneqq\mathbb{E}_{\theta_0}[(\hat\theta_i-\theta_{0,i})(\hat\theta_j-\theta_{0,j})].
\end{align}
We restrict attention to estimators that are locally unbiased at $\theta_0$: $\mathbb{E}_{\theta_0}[\hat\theta_i]=\theta_{0,i}$ and $\partial_j\mathbb{E}_\theta[\hat\theta_i]\bigl|_{\theta=\theta_0}=\delta_{ij}$where $\delta_{ij}$ denotes the Kronecker delta. 

Let $\rho=\sum_k\mu_k\ket{k}\bra{k}$ be the eigenvalue decomposition of $\rho$. The symmetric logarithmic derivative (SLD) quantum Fisher information matrix is
\begin{align}
    \left(\mathcal{F}_\rho^{\SLD}\right)_{ij}\coloneqq\sum_{\substack{k,l\\\mu_k+\mu_l>0}}
    \frac{2\braket{k|\partial_i\rho|l}\braket{l|\partial_j\rho|k}}{\mu_k+\mu_l}.
\end{align}
Throughout the Letter, we assume the local-identifiability condition $\mathcal F_\rho^{\rm SLD}>0$, but do not require $\rho$ to be faithful. 
The QCRB~\cite{helstrom_MinimumMeansquaredErrorestimatesquantum_1967,helstromQuantumDetectionEstimation1969,holevo_ProbabilisticStatisticalAspectsquantumtheory_1982} reads
\begin{align}
    \Sigma_\rho\geq (\mathcal F_\rho^{\SLD})^{-1},\label{eq:CR_single}
\end{align}
where for Hermitian matrices $A$ and $B$, $A\geq B$ means that $A-B$ is positive semidefinite.

In the single-parameter case, the QCRB is scalar and is asymptotically attainable for single-parameter i.i.d. models under suitable regularity conditions~\cite{braunstein_StatisticalDistanceGeometryquantumstates_1994,barndorff-nielsenFisherInformationQuantum2000}.
In multiparameter estimation, by contrast, the matrix inequality need not be simultaneously attainable because of measurement incompatibility~\cite{matsumotoNewApproachCramerRaotype2002,ragyCompatibilityMultiparameterQuantum2016}. This failure of simultaneous attainability motivates the Holevo bound.

\textit{\textbf{Holevo bound.}}---
For an $m$-parameter estimation problem, let $G$ be an $m\times m$ real positive semidefinite weight matrix, which specifies the scalar cost $\Tr(G\Sigma_\rho)$. Let $\mathcal{X}_\rho$ denote the set of Hermitian operators $X=(X_1,\ldots,X_m)$ satisfying
\begin{align}
    \Tr(\rho X_i)=0, \qquad \Tr(X_i \partial_j\rho )=\delta_{ij}.\label{eq:def_unbiased_condition}
\end{align}
The Holevo bound~\cite{holevo_ProbabilisticStatisticalAspectsquantumtheory_1982,nagaoka_newapproachtoCR_1989} is then given by
\begin{align}
    C_\rho^\Holevo(G)\coloneqq \inf_{V\in\mathbb{S}^m}\{\Tr(GV)\colon\exists X\in \mathcal{X}_\rho,\,V\geq Z(X)\}.\label{eq:Holevo_definition}
\end{align}
Here, $\mathbb{S}^m$ denotes the set of all $m\times m$ real symmetric matrices, and $Z(X)$ is an $m\times m$ Hermitian matrix such that $Z(X)_{ij}\coloneqq \Tr(\rho X_iX_j)$. The operator-feasible region of this variational problem is
\begin{align}
    \mathcal{V}_\rho\coloneqq \{V\in\mathbb{S}^m\colon \exists X\in\mathcal{X}_\rho,\,V\geq Z(X)\},\label{eq:def_feasible_region}
\end{align}
with which $C_\rho^\Holevo(G)=\inf_{V\in\mathcal{V}_\rho}\Tr(GV)$ by definition. 
The Holevo bound gives a lower bound on the cost for any locally unbiased measurement~\cite{holevo_ProbabilisticStatisticalAspectsquantumtheory_1982,nagaoka_newapproachtoCR_1989}
\begin{align}
    \Tr( G\Sigma_\rho)\geq C_\rho^\Holevo(G).\label{eq:Holevo_single}
\end{align}

Under suitable regularity conditions, the Holevo bound is asymptotically attainable for i.i.d. models with collective measurements~\cite{hayashi_AsymptoticPerformanceOptimalstateestimation_2008,gutaLocalAsymptoticNormality2006,yangAttainingUltimatePrecision2019,yamagataQuantumLocalAsymptotic2013}. Specifically, since $C_{\rho^{\otimes n}}^\Holevo(G)=\frac{1}{n} C_\rho^\Holevo(G)$~\cite{hayashi_AsymptoticPerformanceOptimalstateestimation_2008}, Eq.~\eqref{eq:Holevo_single} applied for the $n$-copy model gives $n\Tr( G\Sigma_{\rho^{\otimes n}})\geq C_\rho^\Holevo(G)$, and this lower bound is asymptotically tight.

\textit{\textbf{The $\beta$-LD QFI matrices.}}--- 
The $\beta$-LD QFI family, interpolating between the SLD and RLD QFIs, was introduced in Ref.~\cite{yamagataMaximumLogarithmicDerivative2021}. For $\beta\in(-1,1)$, we define
\begin{align}
    \left(\mathcal{F}^\beta_{\rho}\right)_{ij}\coloneqq\sum_{\substack{k,l\\\mu_k+\mu_l>0}}
    \frac{2\braket{k|\partial_i\rho|l}\braket{l|\partial_j\rho|k}}{(1-\beta)\mu_k+(1+\beta)\mu_l}.\label{eq:beta_QFI_def}
\end{align}
At $\beta=0$, it equals the SLD QFI, and for faithful states its $\beta\to1^-$ limit is the RLD QFI. A related family of quantum geometric tensors arises in the resource theory of asymmetry~\cite{yamaguchi_QuantumGeometricTensorDeterminesPureState_2026,yamaguchi_QuantifyingSymmetryBreakingMetricAdjusted_2026}.

Since $(1-\beta)\mu_k+(1+\beta)\mu_l\leq (1+|\beta|)(\mu_k+\mu_l)$, we get $\mathcal{F}^{\beta}_\rho\geq (1+|\beta|)^{-1}\mathcal{F}^{\SLD}_\rho$. Thus $\mathcal{F}_\rho^\beta$ is invertible for every $\beta\in (-1,1)$ under the local-identifiability condition $\mathcal{F}^{\SLD}_\rho>0$. Each $\beta$-LD QFI matrix yields a Cram\'er--Rao-type lower bound on the MSE matrix~\cite{yamagataMaximumLogarithmicDerivative2021}: 
\begin{align}
    \Sigma_\rho\geq\left(\mathcal{F}_\rho^\beta\right)^{-1},\quad \beta\in(-1,1).
\end{align}
The associated scalar bound satisfies~\cite{yamagataMaximumLogarithmicDerivative2021}:
\begin{align}
    C_\rho^\Holevo(G)\geq C_\rho^{(\beta)}(G),\quad  \beta\in(-1,1),\label{eq:Holevo_QFI}
\end{align}
where $C_\rho^{(\beta)}(G)\coloneqq \inf_{V\in\mathbb{S}^m}\{\Tr(GV)\colon V\geq \left(\mathcal{F}_\rho^\beta\right)^{-1}\}$.
Note that $C_\rho^{(\beta)}(G)=C_\rho^{(-\beta)}(G)$ since $(\mathcal{F}^\beta_{\rho})^\top=\mathcal{F}^{-\beta}_{\rho}$. Thus, optimizing over $\beta\in[0,1)$, the maximum logarithmic derivative (MLD) bound $C_\rho^{\text{MLD}}(G)\coloneqq \sup_{\beta\in[0,1)} C_\rho^{(\beta)}(G)$ provides a lower bound on the Holevo bound~\cite{yamagataMaximumLogarithmicDerivative2021}:
\begin{align}
     C_\rho^\Holevo(G)\geq C_\rho^{\text{MLD}}(G).\label{eq:holevo_MLD_general}
\end{align}

\textit{\textbf{Main result.}}--- Our main theorem characterizes the Holevo bound in terms of the $\beta$-QFI family.
\begin{theorem}\label{thm:main_theorem_Holevo_QFI}
    The operator-feasible region in Eq.~\eqref{eq:def_feasible_region} satisfies
    \begin{align}
        \mathcal{V}_\rho =\bigcap_{\beta\in[0,1)}\{V\in\mathbb{S}^m\colon V\geq \left(\mathcal{F}_\rho^\beta\right)^{-1}\}. \label{eq:main_thm_HG_region}
    \end{align}
    Consequently, for any real symmetric weight matrix $G\geq 0$,
    \begin{align}
        C_\rho^\Holevo(G)
        &=\inf_{V\in\mathbb{S}^m}\{\Tr(GV)\colon \forall \beta\in[0,1),\, V\geq \left(\mathcal{F}^\beta_{\rho}\right)^{-1}\}.\label{eq:main_thm_Holevo}
    \end{align}
\end{theorem}
The proof is provided at the end of this Letter. See Figure~\ref{fig:holevo_region}. 
Three remarks are in order. 

\begin{figure}
    \centering
    \includegraphics[width=0.95\linewidth]{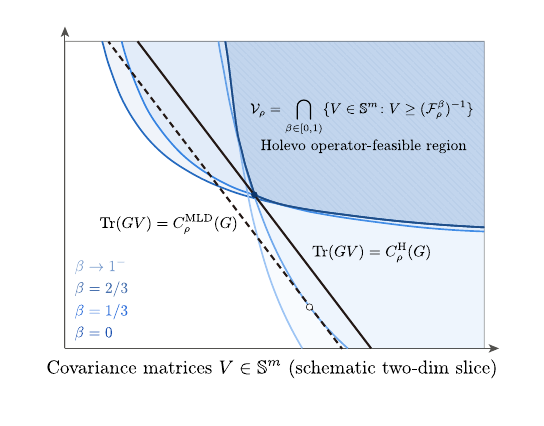}
    \caption{Schematic two-dimensional slice of the space of covariance matrices. Each $\beta$-LD QFI defines a Cram\'er--Rao-type constraint region $\{V\colon V\geq (\mathcal{F}_\rho^\beta)^{-1}\}$ (region above each curve; four members of the family are shown). By Theorem~\ref{thm:main_theorem_Holevo_QFI}, the operator-feasible region $\mathcal{V}_\rho$ of the Holevo bound is the intersection of these regions, $\bigcap_{\beta\in[0,1)}\{V\in\mathbb{S}^m\colon V\geq \left(\mathcal{F}_\rho^\beta\right)^{-1}\}$ (hatched region). For a given weight $G$, the Holevo bound $C^\Holevo_\rho(G)$ is the infimum of $\Tr(GV)$ over $\mathcal{V}_\rho$ (solid line). The MLD bound $C^{\mathrm{MLD}}_\rho$ is the largest single-$\beta$ bound (dashed line). Thus, this figure illustrates the case where $C^{\mathrm{MLD}}_\rho<C^\Holevo_\rho(G)$.}
    \label{fig:holevo_region}
\end{figure}

First, Theorem~\ref{thm:main_theorem_Holevo_QFI} places the apparently different Cram\'er--Rao and Holevo formulations in a common framework. Indeed, in the single-parameter case, Eq.~\eqref{eq:CR_single} is a scalar inequality and can equivalently be written, for a scalar weight $g>0$, as
\begin{align}
    g\Sigma_\rho\geq \inf_{v\in\mathbb{R}}\{gv\colon v\geq (\mathcal{F}_\rho^{0})^{-1}\}
\end{align}
As shown in the Supplemental Material, for a single parameter, $ (\mathcal{F}^0_\rho)^{-1}\geq (\mathcal{F}_\rho^\beta)^{-1}$ for every $\beta$. Hence the SLD constraint is already the strongest member of the $\beta$-LD family appearing in Eq.~\eqref{eq:Holevo_QFI}, and imposing the full family introduces no additional restriction.

Second, Eq.~\eqref{eq:main_thm_HG_region} completely specifies the operator-feasible region $\mathcal{V}_\rho$ without choosing a weight matrix. Weighted Holevo bounds are recovered only afterwards, by taking the infimum of a chosen cost over this common region. Thus, scalarization is not required for a complete information-geometric description of the feasible region in the Holevo optimization. This algebraic characterization does not assert finite-copy attainability of every feasible matrix. 

Third, Theorem~\ref{thm:main_theorem_Holevo_QFI} makes the MLD--Holevo distinction particularly transparent: MLD optimizes over a single $\beta$, whereas Holevo imposes all $\beta$-LD matrix constraints on one common feasible matrix. To make this
distinction explicit, introduce 
\begin{align}
    \Phi_G(V,\beta)\coloneqq 
    \begin{cases}
        \Tr(GV)\qquad &(\text{if }V\geq (\mathcal{F}_\rho^{\beta})^{-1})\\
        +\infty\qquad &(\text{otherwise})
    \end{cases}.
\end{align}
Then the Holevo and MLD bounds can be written, respectively, as
\begin{align}
     C_\rho^{\Holevo}(G)&= \inf_{V\in\mathbb{S}^m}\sup_{\beta\in[0,1)}\Phi_G(V,\beta),\\
     C_\rho^{\text{MLD}}(G)&=\sup_{\beta\in[0,1)} \inf_{V\in\mathbb{S}^m}\Phi_G(V,\beta)
\end{align}
Thus, Eq.~\eqref{eq:holevo_MLD_general} can be understood as a consequence of the minimax inequality.

\textit{\textbf{An explicit separation between the MLD and Holevo bounds.}}--- Before proving Theorem~\ref{thm:main_theorem_Holevo_QFI}, we illustrate the distinction between the MLD bound and the Holevo bound with a simple three-parameter qutrit model. The example makes explicit that different members of the $\beta$-LD QFI family can provide the strongest constraints in different parameter directions. Consequently, optimizing the scalar bound for a single value of $\beta$, as in the MLD construction, need not capture the joint matrix constraints entering the Holevo bound.

Consider $\rho=\diag (p,q,q)$ with $p>q>0$ and $p+2q=1$, in the computational basis
$\{\ket{0},\ket{1},\ket{2}\}$. We define the unitary model $\rho_\theta\coloneqq e^{-\ii \sum_{i=1}^3\theta_i H_i}\rho e^{\ii \sum_{i=1}^3\theta_i H_i}$, where $H_1\coloneqq\ket{0}\bra{1}+\ket{1}\bra{0}$, $H_2\coloneqq -\ii\ket{0}\bra{1}+\ii\ket{1}\bra{0}$ and $H_3\coloneqq\ket{0}\bra{2}+\ket{2}\bra{0}$. 
At $\theta=0$, a direct calculation from Eq.~\eqref{eq:beta_QFI_def} gives
\begin{align}
    \mathcal{F}_\rho^\beta=\frac{\kappa}{1-\beta^2r^2}
    \begin{pmatrix}
        1&\ii\beta r&0\\
        -\ii\beta r&1&0\\
        0&0&1
    \end{pmatrix},
\end{align}
where $r\coloneqq(p-q)/(p+q)\in(0,1)$ and $\kappa \coloneqq (4(p-q)^2)/(p+q)$. 
In particular, $\mathcal F_\rho^{0}=\kappa I_3>0$. Its inverse is
\begin{align}
    \left(\mathcal{F}_\rho^\beta\right)^{-1}=\frac{1}{\kappa}
    \begin{pmatrix}
        1&-\ii\beta r&0\\
        \ii\beta r&1&0\\
        0&0&1-\beta^2r^2
    \end{pmatrix}.
\end{align}

We choose the weight $G=\kappa I_3$, where the factor $\kappa$ merely removes an overall scale. To see the structure of the matrix constraint, consider $v_1\coloneqq\frac{1}{\sqrt2}(1,\ii,0)^\top$ and $v_2\coloneqq(0,0,1)^\top$. 
For any real symmetric $V\in\mathbb S^3$ satisfying $V\geq(\mathcal F_\rho^\beta)^{-1}$, the condition along $v_1$ gives
\begin{align}
    \kappa(V_{11}+V_{22}) \geq 2+2\beta r,\label{eq:constraint_v1}
\end{align}
while the condition along $v_2$ gives
\begin{align}
    \kappa V_{33} \geq 1-\beta^2r^2. \label{eq:constraint_v2}
\end{align}
These two constraints have opposite dependences on $\beta$: Eq.~\eqref{eq:constraint_v1} becomes stronger as $\beta$ increases and is strongest in the RLD limit $\beta\to1^-$, whereas Eq.~\eqref{eq:constraint_v2} is strongest at the SLD point $\beta=0$. Thus, no single value of $\beta$ simultaneously gives the strongest constraint in all parameter directions. This competition is the origin of the separation between the MLD and Holevo bounds in this example.

Let us first evaluate the MLD bound. For a fixed $\beta$, adding Eqs.~\eqref{eq:constraint_v1} and~\eqref{eq:constraint_v2} yields
\begin{align}
    \Tr(GV)\geq 3+2\beta r-\beta^2r^2.
\end{align}
This lower bound is attained in the fixed-$\beta$ optimization. Indeed, defining $V_\beta\coloneqq\frac{1}{\kappa}\diag (1+\beta r, 1+\beta r, 1-\beta^2r^2)$, we have
\begin{align}
    V_\beta-\left(\mathcal F_\rho^\beta\right)^{-1}=\frac{1}{\kappa}
    \begin{pmatrix}
        \beta r&\ii\beta r&0\\
        -\ii\beta r&\beta r&0\\
        0&0&0
    \end{pmatrix}
    \geq0
\end{align}
for $\beta\in[0,1)$, and $\Tr( G V_\beta)=3+2\beta r-\beta^2r^2$. Therefore,
\begin{align}
    C_\rho^{(\beta)}(G) = 3+2\beta r-\beta^2r^2. \label{eq:qutrit_beta_bound}
\end{align}
Since $0<r<1$, this expression is strictly increasing for $\beta\in[0,1)$, and hence
\begin{align}
    C_\rho^{\mathrm{MLD}}(G)&=\sup_{\beta\in[0,1)}C_\rho^{(\beta)}(G)=3+2r-r^2.\label{eq:qutrit_mld}
\end{align}

The important point is that the minimization defining $C_\rho^{(\beta)}(G)$ is performed separately for each $\beta$. Accordingly, the optimizing matrix $V_\beta$ is allowed to depend on $\beta$. The MLD bound then retains only the largest scalar optimum among these independently optimized problems; it does not require one common feasible matrix to satisfy the constraints associated with different values of $\beta$ simultaneously. For example, in the RLD limit, the optimal matrix has
\begin{align}
    \lim_{\beta\to 1^-}\kappa(V_{\beta})_{33}=1-r^2<1,
\end{align}
and therefore violates Eq.~\eqref{eq:constraint_v2} at the SLD point $(\beta=0)$: $\kappa V_{33}\geq1$.

The situation is different for the Holevo bound. By Theorem~\ref{thm:main_theorem_Holevo_QFI}, a single matrix $V$ must satisfy
\begin{align}
    V\geq(\mathcal F_\rho^\beta)^{-1} \qquad\text{for every }\beta\in[0,1).\label{eq:qutrit_joint_all}
\end{align}
Hence Eq.~\eqref{eq:constraint_v1} in the limit $\beta\to1^-$ requires
\begin{align}
    \kappa(V_{11}+V_{22})\geq 2+2r,
\end{align}
while Eq.~\eqref{eq:constraint_v2} at $\beta=0$ simultaneously requires
\begin{align}
    \kappa V_{33}\geq 1.
\end{align}
Consequently, Eq.~\eqref{eq:qutrit_joint_all} requires
\begin{align}
    \Tr(GV)\geq 3+2r.
\end{align}
This bound is attained in the joint-constraint optimization by $V_* \coloneqq\frac{1}{\kappa}\diag(1+r , 1+r , 1)$. Indeed,
\begin{align}
    V_*-\left(\mathcal F_\rho^\beta\right)^{-1}=\frac{1}{\kappa}
    \begin{pmatrix}
        r&\ii\beta r&0\\
        -\ii\beta r&r&0\\
        0&0&\beta^2r^2
    \end{pmatrix}
    \geq0
\end{align}
for every $\beta\in[0,1)$, and $\Tr(GV_*)=3+2r$. Therefore,
\begin{align}
    C_\rho^\Holevo(G)=3+2r. \label{eq:qutrit_holevo}
\end{align}
Combining Eqs.~\eqref{eq:qutrit_mld} and~\eqref{eq:qutrit_holevo},
we obtain the strict separation
\begin{align}
    C_\rho^\Holevo(G)-C_\rho^{\mathrm{MLD}}(G)=r^2>0.
\end{align}

The origin of the gap can be seen particularly clearly by separating the two contributions to the cost. For a fixed $\beta$, define $a(\beta)\coloneqq 2+2\beta r$ and $ b(\beta)\coloneqq 1-\beta^2r^2$.  The MLD bound optimizes the sum for one common value of $\beta$,
\begin{align}
    C_\rho^{\mathrm{MLD}}(G)=\sup_{\beta\in[0,1)}\bigl[a(\beta)+b(\beta)\bigr].
\end{align}
For this model, the jointly feasible matrix $V_*$ also shows that both contributions can reach their separate suprema simultaneously:
\begin{align}
    C_\rho^\Holevo(G)=\sup_{\beta\in[0,1)}a(\beta)+\sup_{\beta\in[0,1)}b(\beta).
\end{align}
Here the first supremum is obtained in the RLD limit $\beta\to 1^-$, while the second is obtained at the SLD point $\beta=0$, producing the gap $C_\rho^{\mathrm{MLD}}(G)<C_\rho^\Holevo(G)$. This competition illustrates that the
full QFI-family characterization of the Holevo bound in Theorem~\ref{thm:main_theorem_Holevo_QFI} is genuinely stronger than the maximization of the corresponding scalar Cram\'er--Rao-type bounds at the same weight $G$.

\textit{\textbf{Proof of Theorem~\ref{thm:main_theorem_Holevo_QFI}.}}--- 
We reformulate the operator-feasibility condition defining $\mathcal{V}_\rho$ in Eq.~\eqref{eq:def_feasible_region}.
We introduce a real-coefficient expression of the feasibility condition, analogous to Refs.~\cite{hayashi_AsymptoticPerformanceOptimalstateestimation_2008,albarelli_EvaluatingHolevoCramerRaoBoundMultiparameter_2019}.
Let $\mathcal{X}_\rho'$ denote the set of Hermitian operators $X'=(X'_1,\ldots,X'_m)$ satisfying $\Tr(X_i'\partial_j\rho)=\delta_{ij}$. For any $X'\in\mathcal{X}_\rho'$, $X_i\coloneqq X'_i-\Tr(\rho X_i')I$ satisfies $\Tr(\rho X_i)=0$ and $\Tr(X_i\partial_j\rho)=\Tr(X_i'\partial_j\rho)=\delta_{ij}$ since $\Tr(\partial_j\rho)=0$, and therefore $X=(X_1,\ldots,X_m)\in\mathcal{X}_\rho$. Since $(Z(X))_{ij}=(Z(X'))_{ij}-\Tr(\rho X'_i)\Tr(\rho X'_j)$, we obtain $Z(X')\geq Z(X)$. Thus, since $\mathcal{X}_\rho\subset \mathcal{X}_\rho'$, removing the condition $\Tr(\rho X_i)=0$ does not
change the operator-feasible region. 
Moreover, we may restrict $X'\in \mathcal{X}'_\rho$ to the real vector space $\mathcal{S}\coloneqq \{A=A^\dag\colon QAQ=0\}$, where $Q$ is the projector onto the kernel of $\rho$. Indeed, for $X'\in\mathcal{X}'_\rho$, define $\tilde{X}'=(\tilde{X}_1',\ldots,\tilde{X}_m')$ with $\tilde{X}_i'\coloneqq X_i'-QX_i'Q$. Since $Q(\partial_j \rho)Q=0$~\cite{yamaguchi_QuantifyingSymmetryBreakingMetricAdjusted_2026}, we have $\tilde{X}'\in\mathcal{X}'_\rho$. Also, $Q\rho=\rho Q=0$ implies $Z(X')=Z(\tilde{X}')$. Hence, without loss of generality, we may restrict each component $X_i'$ of $X'\in \mathcal{X}'_\rho$ to the real vector space $\mathcal{S}$. Thus, we obtain
\begin{align}
    &\exists X\in\mathcal{X}_\rho , \quad V\geq Z(X)\nonumber\\
    &\Longleftrightarrow \exists X'\in\mathcal{X}_\rho',\, X_i'\in\mathcal{S},\,V\geq Z(X').\label{eq:reduction}
\end{align}

Let $\{S_\alpha\}_{\alpha=1}^D$ be an orthonormal basis of $\mathcal{S}$ with respect to the Hilbert--Schmidt inner product, where $D\coloneqq \dim \mathcal{S}$. Since we may take $X_i'\in \mathcal{S}$, it can be expanded as $X_i'=\sum_{\alpha=1}^DY_{i\alpha}S_\alpha$ with real coefficients $Y_{i\alpha}$, which defines a real matrix $Y\in\mathbb{R}^{m\times D}$. Introducing matrices $T\in\mathbb{R}^{D\times m}$ and $M\in\mathbb{C}^{D\times D}$ by
\begin{align}
    T_{\alpha j}\coloneqq \Tr(S_\alpha\partial_j \rho),\quad M_{\alpha\gamma}\coloneqq \Tr(\rho S_\alpha S_\gamma),\label{eq:definition_T_and_M}
\end{align}
we have
\begin{align}
    &\exists X'\in\mathcal{X}_\rho',\, X_i'\in\mathcal{S},\,V\geq Z(X')\nonumber\\
    &\Longleftrightarrow\exists Y\in\mathbb{R}^{m\times D},\, YT=I_m,\,V\geq YMY^\top,\label{eq:SDP_form}
\end{align}
which is the real-coefficient form underlying the finite semidefinite-programming formulation of the Holevo bound in Eq.~(11) of Ref.~\cite{albarelli_EvaluatingHolevoCramerRaoBoundMultiparameter_2019}.

Now, we translate the above expression into a condition using the $\beta$-LD QFI family. 
From the definition in Eq.~\eqref{eq:definition_T_and_M}, $M=M^\dag \geq 0$.  Since $M$ is Hermitian, we decompose it into its real and imaginary parts by 
\begin{align}
    K\coloneqq (M+M^\top)/2,\qquad \Omega\coloneqq (M-M^\top)/(2\ii),
\end{align}
which satisfies $K^\top =K$ and $\Omega^\top=-\Omega$. 
We here prove $K>0$. For any $x\in\mathbb{R}^D$, $x^\top K x=\|B\rho^{1/2}\|_{\HS}^2$, where $B\coloneqq \sum_{\alpha=1}^Dx_\alpha S_\alpha$. Therefore, if $x^\top K x=0$, then $B\rho^{1/2}=0$ and hence $BP=0$, where $P$ is the projector onto the support of $\rho$. Since $B=B^\dag$, this also implies $PB=0$. Since $B\in\mathcal{S}$, we also have $QBQ=0$; together with $BP=PB=0$, this gives $B=0$, which in turn implies $x=0$ since $\{S_\alpha\}_\alpha$ is linearly independent.  Thus, $K>0$. 

We next observe that every feasible matrix $V$ is automatically positive definite. Let $(V,Y)$ be any feasible pair satisfying the condition in Eq.~\eqref{eq:SDP_form}. Let $u\in\mathbb{R}^m$ be any nonzero real vector. The condition $YT=I_m$ implies $Y^\top u\neq 0$, and $Y\in\mathbb{R}^{m\times D}$ implies $Y^\top u\in\mathbb{R}^D$. Evaluating the matrix inequality $V\geq YMY^\top$ on the real vector $u$ gives
\begin{align}
    u^\top V u&\geq u^{\top} Y M Y^\top u\nonumber \\
    &= u^{\top} Y K Y^\top u=(Y^\top u)^\top K(Y^\top u)>0,
\end{align}
where we have used $M=K+\ii\Omega$, $\Omega^\top =-\Omega$, and $K>0$. Consequently, $V>0$. 

These conditions $K=K^\top>0$, $\Omega^\top=-\Omega$, $M=K+\ii\Omega\geq 0$, together with $T\in\mathbb{R}^{D\times m}$ and the automatic positivity $V>0$, make Theorem H.8 of Ref.~\cite{yamaguchi_QuantifyingSymmetryBreakingMetricAdjusted_2026} applicable. Concretely, it shows
\begin{align}
    &\exists Y\in\mathbb{R}^{m\times D},\, YT=I_m,\,V\geq YMY^\top\nonumber \\
    &\Longleftrightarrow V>0,\,\forall\beta\in(-1,1),\,T^\top (K-\ii\beta \Omega)^{-1}T\geq V^{-1}.\label{eq:inversion}
\end{align}
We prove that $T^\top (K-\ii\beta \Omega)^{-1}T=\mathcal{F}_\rho^\beta$. 
Let $\mathcal{S}_{\mathbb{C}}$ be the complexification of $\mathcal{S}$.
Define the right- and left-multiplication superoperators $\mathbb{R}_\rho,\mathbb{L}_\rho:\mathcal{S}_{\mathbb{C}}\to \mathcal{S}_{\mathbb{C}}$ by
\begin{align}
    \mathbb{R}_\rho(A)\coloneqq A\rho,\qquad \mathbb{L}_\rho(A)\coloneqq \rho A.
\end{align}
Their components in the basis $\{S_\alpha\}_{\alpha=1}^D$ are given by
\begin{align}
    (\mathbb{R}_\rho)_{\alpha\gamma}&\coloneqq \Tr(S_\alpha \mathbb{R}_\rho(S_\gamma))=M_{\alpha\gamma},\\
    (\mathbb{L}_\rho)_{\alpha\gamma}&\coloneqq \Tr(S_\alpha \mathbb{L}_\rho(S_\gamma))=M_{\gamma\alpha}.
\end{align}
Therefore, the superoperator
\begin{align}
    \mathbb{J}_\rho^\beta\coloneqq\frac{1-\beta}{2}\mathbb{R}_\rho+\frac{1+\beta}{2}\mathbb{L}_\rho
\end{align}
satisfies $(\mathbb{J}_\rho^\beta)_{\alpha\gamma}=(K-\ii \beta \Omega)_{\alpha\gamma}$. Using the eigenvalue decomposition $\rho=\sum_k\mu_k\ket{k}\bra{k}$, for any matrix element $\ket{l}\bra{k}$ with $\mu_k+\mu_l>0$, we have
\begin{align}
    \mathbb{J}_\rho^\beta\left(\ket{l}\bra{k}\right)=\frac{(1-\beta)\mu_k+(1+\beta)\mu_l}{2}\ket{l}\bra{k}, 
\end{align}
and hence $\mathbb{J}_\rho^\beta>0$ on $\mathcal{S}_{\mathbb{C}}$. Since $\partial_i\rho=\sum_{\alpha=1}^D T_{\alpha i}S_\alpha$, for every $\beta\in(-1,1)$, Eq.~\eqref{eq:beta_QFI_def} gives
\begin{align}
    (\mathcal{F}_\rho^\beta)_{ij}&=\Tr(\partial_i\rho (\mathbb{J}_\rho^\beta)^{-1}(\partial_j\rho))\\
    &=\sum_{\alpha,\gamma=1}^{D}T_{\alpha i}((K-\ii\beta \Omega)^{-1})_{\alpha\gamma}T_{\gamma j}\\
    &=\left(T^\top (K-\ii \beta \Omega)^{-1}T\right)_{ij},
\end{align}
i.e., $\mathcal{F}_\rho^\beta=T^\top (K-\ii \beta \Omega)^{-1}T$. Thus, the definition of $\mathcal{V}_\rho$ and Eqs.~\eqref{eq:reduction},~\eqref{eq:SDP_form}, and~\eqref{eq:inversion} imply
\begin{align}
    V\in\mathcal{V}_\rho
    &\Longleftrightarrow\exists X\in\mathcal{X}_\rho',\,V\geq Z(X)\nonumber\\
    &\Longleftrightarrow V>0,\,\forall\beta\in(-1,1),\mathcal{F}_\rho^\beta\geq V^{-1}\nonumber \\
    &\Longleftrightarrow V>0,\,\forall\beta\in(-1,1),\, V\geq (\mathcal{F}_\rho^\beta)^{-1}.
\end{align}
Here, we used $A\geq B\Leftrightarrow B^{-1}\geq A^{-1}$ for positive definite matrices $A,B$. Since $V$ is real symmetric and $\mathcal{F}_\rho^{-\beta}=(\mathcal{F}_\rho^\beta)^\top$, the constraints for $\beta\in(-1,1)$ are equivalent to those for $\beta\in[0,1)$. Moreover, the $\beta=0$ constraint implies $V\geq (\mathcal{F}_\rho^{\SLD})^{-1}>0$, so the separate condition $V>0$ is redundant. This completes the proof of Eq.~\eqref{eq:main_thm_HG_region}. 
Taking the infimum of $\Tr(GV)$ over the same feasible set yields Eq.~\eqref{eq:main_thm_Holevo} for any $G\geq0$, including singular weights, since the above feasible-set identity is independent of $G$

\begin{acknowledgments}
K.Y. acknowledges support through the Dieter Schwarz Foundation, a Discovery Grant of the Natural Sciences and Engineering Research Council of Canada (NSERC), and is grateful for the hospitality of Perimeter Institute, where part of this work was carried out. Research at Perimeter Institute is supported in part by the Government of Canada through the Department of Innovation, Science and Economic Development Canada and by the Province of Ontario through the Ministry of Colleges, Universities, Research Excellence and Security.
H.T. was supported by JSPS Grants-in-Aid for Scientific Research No. JP25K00924, MEXT KAKENHI Grant-in-Aid for Transformative
Research Areas B ``Quantum Energy Innovation” Grant Numbers 24H00830 and 24H00831, JST FOREST No. JPMJFR2365, JST MOONSHOT No. JPMJMS256E and Royal Society International Collaboration Awards 2025 Flexigrant number ICA/R2/252240.

\end{acknowledgments}


%

\appendix
\widetext

\clearpage

\setcounter{page}{1}

\renewcommand{\theequation}{S.\arabic{equation}}
\setcounter{equation}{0}

\begin{center}
{\large \bf Supplemental Material for\\
\makebox[\textwidth]{``Exact Characterization of the Holevo Bound by a Quantum Fisher Information Family''}}\\
\vspace*{0.3cm}
Koji Yamaguchi$^{1,2}$ and Hiroyasu Tajima$^{3,4}$\\
\vspace*{0.1cm}

$^{1}${\small \it Department of Physics, University of Waterloo, Waterloo, ON N2L 3G1, Canada}
\\
$^{2}${\small \it Perimeter Institute for Theoretical Physics, Waterloo, Ontario N2L 2Y5, Canada}

$^{3}${\small \it Department of Informatics, Faculty of Information Science and Electrical Engineering,
Kyushu University, 744 Motooka, Nishi-ku, Fukuoka, 819-0395, Japan}

$^{4}${\small \it JST, FOREST, 4-1-8 Honcho, Kawaguchi, Saitama, 332-0012, Japan}

\end{center}

\vspace{1cm}

We prove the scalar comparison used in the first remark after Theorem~\ref{thm:main_theorem_Holevo_QFI} for the single-parameter case, i.e., $m=1$. Since $\partial_i\rho$ is Hermitian, $\overline{\braket{k|\partial_i\rho |l}}=\braket{l|\partial_i\rho|k}$. Therefore,
\begin{align}
    \mathcal{F}_\rho^\beta=\sum_{k;\mu_k>0}\frac{|\braket{k|\partial_1\rho|k}|^2}{\mu_k}+\sum_{\substack{k<l\\\mu_k+\mu_l>0}}\frac{4(\mu_k+\mu_l)|\braket{k|\partial_1\rho|l}|^2}{(\mu_k+\mu_l)^2-\beta^2(\mu_k-\mu_l)^2}.
\end{align}
For $|\beta|<1$, every denominator in the second sum is positive and is no larger than its value at $\beta=0$, while the first sum is independent of $\beta$. Hence, 
\begin{align}
    \mathcal{F}_\rho^\beta\geq \mathcal{F}_\rho^0>0,\qquad \left(\mathcal{F}_\rho^\beta\right)^{-1}\leq \left(\mathcal{F}_\rho^0\right)^{-1}.
\end{align}
Consequently, for the single-parameter case,
\begin{align}
    \bigcap_{\beta\in[0,1)}\{v\in\mathbb{R}\colon v\geq  \left(\mathcal{F}_\rho^\beta\right)^{-1}\}=\{v\in\mathbb{R}\colon v\geq  \left(\mathcal{F}_\rho^0\right)^{-1}\}.
\end{align}
Thus, the SLD constraint alone determines the common feasible region for a single parameter. Substitution into Eq.~\eqref{eq:main_thm_Holevo} gives $C_\rho^\Holevo(g)=g/\mathcal{F}_\rho^0$ for any scalar weight $g\geq 0$.

\end{document}